\documentclass[12pt,onecolumn]{IEEEtran}

\usepackage[lmargin=0.75in,rmargin=0.75in,tmargin=.75in,bmargin=0.7in]{geometry}

\usepackage{amsmath}
\usepackage{amssymb}
\usepackage{amsfonts}
\usepackage{epsfig}
\usepackage{subfigure}
\usepackage{calc}
\usepackage{color}
\usepackage[all]{xy}
\usepackage{bm}
\usepackage{hyperref}
\title{Optimal Linear Codes with a Local-Error-Correction Property}

\author{N. Prakash, Govinda M. Kamath, V. Lalitha and P. Vijay Kumar \\
Dept. of ECE, Indian Institute of Science, Bangalore - 560012, India\\
email: \{prakashn, govinda, lalitha, vijay\}@ece.iisc.ernet.in.}

\newtheorem{defn}{Definition}
\newtheorem{thm}{Theorem}
\newtheorem{prop}[thm]{Proposition}
\newtheorem{lem}[thm]{Lemma}
\newtheorem{cor}[thm]{Corollary}
\newtheorem{note}{Remark}

\newtheorem{conj}[thm]{Conjecture}
\newtheorem{constr}[thm]{Construction}

\newcommand{\bit}{\begin{itemize}}
\newcommand{\eit}{\end{itemize}}
\newcommand{\bcor}{\begin{cor}}
\newcommand{\ecor}{\end{cor}}
\newcommand{\beq}{\begin{equation}}
\newcommand{\eeq}{\end{equation}}
\newcommand{\beqn}{\begin{equation*}}
\newcommand{\eeqn}{\end{equation*}}
\newcommand{\bea}{\begin{eqnarray}}
\newcommand{\eea}{\end{eqnarray}}
\newcommand{\bean}{\begin{eqnarray*}}
\newcommand{\eean}{\end{eqnarray*}}
\newcommand{\ben}{\begin{enumerate}}
\newcommand{\een}{\end{enumerate}}
\newcommand{\bdefn}{\begin{defn}}
\newcommand{\edefn}{\end{defn}}
\newcommand{\bnote}{\begin{note}}
\newcommand{\enote}{\end{note}}
\newcommand{\bprop}{\begin{prop}}
\newcommand{\eprop}{\end{prop}}
\newcommand{\blem}{\begin{lem}}
\newcommand{\elem}{\end{lem}}
\newcommand{\bthm}{\begin{thm}}
\newcommand{\ethm}{\end{thm}}
\newcommand{\bconj}{\begin{conj}}
\newcommand{\econj}{\end{conj}}
\newcommand{\bconstr}{\begin{constr}}
\newcommand{\econstr}{\end{constr}}
\newcommand{\bpf}{\begin{proof}}
\newcommand{\epf}{\end{proof}}

\begin{document}

\maketitle

\begin{abstract}
\normalsize
Motivated by applications to distributed storage, Gopalan \textit{et al} recently introduced the interesting notion of information-symbol locality in a linear code. By this it is meant that each message symbol appears in a parity-check equation associated with small Hamming weight,
thereby enabling recovery of the message symbol by examining a small number of other code symbols.  This notion is expanded to
the case when all code symbols, not just the message symbols, are covered by such ``local'' parity.
In this paper, we extend the results of Gopalan et. al. so as to permit recovery of an erased code symbol even in the presence of errors in local parity symbols.
We present tight bounds on the minimum distance of such codes and exhibit codes that are optimal with respect to the local error-correction property.  As a corollary, we obtain an upper bound on the minimum distance of a concatenated code.

\end{abstract}

\section{Introduction}

In \cite{GopHuaSimYek}, Gopalan \textit{et al} introduced the interesting and practically relevant notion of locality of information.  The $i^{th}$ code-symbol $c_i, \ 1 \leq i \leq n$, of an $[n, k, d]$ linear code $\mathcal{C}$ over the field $\mathbb{F}_q$ is said to have \textit{locality} $r$ if this symbol can be recovered by accessing at most $r$ other code symbols of code $\mathcal{C}$.  Equivalently, for any coordinate $i$, there exists a row in the parity-check matrix of the code of Hamming weight at most $r+1$, whose support includes $i$.   An $(r,d)$ code was defined as a systematic linear code $\mathcal{C}$ having minimum distance $d$, where all $k$ message symbols have locality $r$. It was shown that the minimum distance of an $(r,d)$ code is upper bounded by
\begin{equation} \label{eq:bound_gopalan}
d \ \leq \ n - k - \left\lceil{\frac{k}{r}}\right\rceil + 2.
\end{equation}
A class of codes constructed earlier and known as pyramid codes~\cite{HuaCheLi} are shown to be $(r, d)$ codes that are optimal with respect to this bound.

The concept of an $(r, d)$ code was motivated by the problem of designing efficient codes for the distributed storage of data across nodes in a network. Since nodes are prone to failure, there is need to protect the data using an error-correcting code.  A second important requirement in this setting, is the ability to efficiently bring up a failed node. Here, $(r, d)$ codes offer the advantage that in the event of a single node failure, the node can be \textit{locally} recovered by connecting to at most $r$ other nodes.


A natural extension to the concept of an $(r,d)$ code, is a code that would allow local recovery of a failed node, even in the presence of failures in other nodes of the network.  Multiple node failures are not uncommon in distributed data storage, and a number of coding schemes for tolerating such multiple node failures exist in practice \cite{HuaCheLi}\cite{BlaBraBruMen}\cite{CorEngGoeGrcKleLeoSan}.  This motivates the definition of the class of $(r, d, \delta)$ local-error-correction (LEC) codes given below.

\vspace{0.1in}

\begin{defn}
The $i$th code symbol $c_i, 1 \leq i \leq n$, in an $[n, k, d]$ linear code $\mathcal{C}$, will be said to have \textit{locality} $(r, \delta)$ if there exists a punctured subcode of ${\cal C}$ with support containing $i$, whose length is at most $r + \delta - 1$, and whose minimum distance is at least $\delta$.  Equivalently, there exists a subset $S_i \subseteq [n] = \{1, \ldots, n\}$ such that
\begin{itemize}
\item  $i \in S_i$ and $|S_i| \leq r + \delta - 1$,
\item  the minimum distance of the code ${\cal C} |_{S_i}$ obtained by deleting code symbols $c_i, \ i \in [n]\backslash S_i$, is least $\delta$.
\end{itemize}
\end{defn}

\vspace{0.1in}

Since the dual of a punctured code is a shortened code, this also implies that we may regard the parity-check matrix $H$ of the code as containing for some $\nu_i, 1 \leq \nu_i \leq n-k$, a $(\nu_i \times n)$ submatrix $H_i$ having rank $\nu_i$, support $S_i$, and the property that any $\delta-1$ columns of $H_i$ with indices drawn from $S_i$, are linearly independent.

A systematic $[n,k,d]$ linear code $\mathcal{C}$ will be said to be an $(r, \delta)_i$ code, if all $k$ message (or information) symbols have locality $(r, \delta)$. We will also refer to such a code as having \textit{information locality} $(r, \delta)$. It is clear that if we employ an $(r, \delta)_i$ code for the distributed storage of data, a systematic node can be locally repaired by connecting to $r$ other nodes, even if $\delta - 2$ other nodes fail. An additional advantage of an $(r, \delta)_i$ code is that even when the other nodes are intact, the code provides multiple options for locally repairing a failed systematic node, which in a network setting, can be used to balance traffic across the network \footnote{By connecting to any $r$ out of the $r+\delta-2$ nodes which locally protect the failed node, one can recover the failed node.}.   The $(r, d)$ codes introduced by Gopalan \textit{et al} correspond to $(r, 2)_i$ codes in the present notation.

By using properties of the generalized Hamming weights~\cite{Wei} of a code (also known as minimum support weights\cite{HelKloLevYtr}), we will show that the minimum distance of an $(r, \delta)_i$ is upper bounded (Theorem~\ref{thm:information_locality}) by
\begin{equation} \label{eq:bound_locality}
d \ \leq \ n - k + 1 - \left(\left\lceil{\frac{k}{r}}\right\rceil - 1\right)(\delta - 1).
\end{equation}
As was the case with the $(r,d)$ codes introduced in \cite{GopHuaSimYek}, a class of pyramid codes turns out to provide examples of optimal $(r, \delta)_i$ codes, i.e., $(r,\delta)_i$ codes is which the bound in \eqref{eq:bound_locality} is achieved with equality. For the special case when $r|k$, we will identify conditions that the parity check matrix of an optimal $(r, \delta)_i$ code must necessarily satisfy.

We will term a code in which all the $n$ symbols of an $[n,k,d]$ code have locality $(r, \delta)$ as codes as having \textit{all-symbol locality} $(r, \delta)$ and denote such codes as $(r,\delta)_a$ codes.  Thus, whenever we speak of either an $(r,\delta)_i$ or else an $(r,\delta)_a$ code, it will be assumed that the length, dimension and minimum distance of the linear code are understood from the context and are typically denoted by $n,k,d$ respectively.   Clearly, codes with all-symbol locality are a subset of the set of codes with just information locality.  Nevertheless, it turns out that when $(r + \delta -1)|n$, one can show the existence of codes with all-symbol locality $(r, \delta)$, which satisfy the upper bound on minimum distance given in \eqref{eq:bound_locality}.  We will also present an explicit code having all-symbol locality, for the case when the code length $n$ is of the form $n = \left\lceil \frac{k}{r}\right\rceil(r+\delta-1)$.

Through out this write up, we will assume without loss of generality, that the $[n, k, d]$ code $\mathcal{C}$ under study, is systematic, with information symbols present in the first $k$ coordinates. For a codeword $\bf{c} \in \mathcal{C}$, we will use $\text{supp}({\bf{c}})$ to denote the support $\{i \in [n] \left| \ c_i \neq 0  \right. \}$ of the codeword.  The support of a subcode $\mathcal{D}$ of $\mathcal{C}$, is defined by $\text{supp}({\mathcal{D}}) \triangleq \cup_{\bf{c} \in \mathcal{D}} \text{supp}(\bf{c})$. For a set $S \subset [n]$, we will use $\mathcal{C}|_{S}$ and $\mathcal{C}^{S}$ to denote respectively, the punctured and shortened codes of $C$ associated with the coordinate set $S$.  By this we mean that under either the puncturing or shortening operation, the coordinates of the code lying in $[n]\backslash S$ are suppressed. Also, for any set $S$, the cardinality of the set will be denoted by $|S|$.

Section \ref{sec:prelims} presents background on generalized Hamming weights, while codes with information and all-symbol locality are treated in Sections \ref{sec:info_locality} and \ref{sec:all_symbol_locality} respectively.  In the final section, Section~\ref{sec:concatenated_codes}, we present as a corollary, an upper bound on the minimum distance of a concatenated code.

\section{Generalized Hamming weights} \label{sec:prelims}
In this section, we review the definition of the generalized Hamming weight (GHW) of a code~\cite{Wei,HufPle} and see how the GHWs of a code are related to those of its dual.  We introduce the notion of a \textit{gap} which will play an important role in our subsequent proofs.


\vspace{0.1in}

\begin{defn}
The $i^{th}$, $1 \leq i \leq k$, generalized Hamming weight of a code ${\cal C}$ is defined by
\begin{equation}
d_i({\cal C}) \ = \ d_i \ = \  \min_{\substack{ \mathcal{D} < \mathcal{C} \\ \text{dim}(\mathcal{D}) = i }} \left|\text{Supp}({\cal D}) \right| ,
\end{equation}
where $\mathcal{D} < \mathcal{C}$, is used to denote a subcode $\mathcal{D}$ of $\mathcal{C}$.
\end{defn}

\vspace{0.1in}

It is well known that
\begin{equation}
d=d_1 < d_2 < \ldots < d_k = n.
\end{equation}
We will call the complement of the set $\{d_i, 1 \leq i \leq k\}$, in $[n]$, as the set of \textit{gap numbers} (more simply, gaps) of the code ${\cal C}$ and denote them by the set $\{ g_i, \ 1 \leq i \leq n-k \}$, where
\begin{equation}
\{ g_i, \ 1 \leq i \leq n-k \} \ = \ [n] \setminus \{ d_i, \ 1 \leq i \leq k \}.
\end{equation}
Similarly, let the sets $\{d_j^{\perp}, \ 1 \leq j \leq n-k \}$ and $\{ g_i^{\perp}, \ 1 \leq i \leq k \}$ respectively denote the GHWs and gaps of the dual code ${\cal C}^\perp$. The following lemma~\cite{Wei} relates the GHWs of $\mathcal{C}$ to those of $\mathcal{C}^{\perp}$.

\vspace{0.1in}

\begin{lem} \label{lem:GHW_code_dual_relation}
 \begin{equation} \label{eq:GHW_code_dual_relation}
 \{d_i, \ 1 \leq i \leq k\} \ = \ [n] \setminus \{n + 1 - d_j^{\perp}, \ \ 1 \leq j \leq n-k \}.
 \end{equation}
\end{lem}

\vspace{0.1in}

In terms of the gaps of the dual code ${\cal C}^\perp$, \eqref{eq:GHW_code_dual_relation} can be rewritten as
\begin{equation} \label{eq:GHW_code_dual_relation_gaps}
d_i = (n+1)-g^{\perp}_{k-i+1} , \ \ 1 \leq i \leq k.
\end{equation}
In particular, the minimum distance $d$ of $\mathcal{C}$ and the largest gap $g^{\perp}_{k}$ of $\mathcal{C}^{\perp}$
are related by
\begin{equation} \label{eq:GHW_dmin_largestgap}
d \ = \ d_1 \ = \ (n+1)- g^{\perp}_{k}.
\end{equation}
This relation will be used to derive an upper bound on the minimum distance of $(r,\delta)_i$ codes.

\section{Codes with information locality} \label{sec:info_locality}

In this section, Theorem \ref{thm:information_locality} will establish the upper bound appearing in \eqref{eq:bound_locality}, on the minimum distance of $(r,\delta)_i$ codes.  It will then be shown that pyramid codes, under an appropriate choice of parameters, are optimal with respect to this bound.   Necessary conditions for optimality of an $(r,\delta)_i$ code for the case when $r|k$, are identified in Theorem \ref{thm:necessary_conditions_info_locality}.

\vspace{0.1in}

\begin{thm} \label{thm:information_locality}
The minimum distance $d$ of an $(r,\delta)_i$ code ${\cal C}$  is upper bounded by
\begin{eqnarray} \label{eq:bound_info_locality}
d \ \leq \ n - k + 1 - \left(\left\lceil{\frac{k}{r}}\right\rceil - 1\right)(\delta - 1).
\end{eqnarray}
\end{thm}

\begin{proof}
From \eqref{eq:GHW_dmin_largestgap}, the minimum distance of ${\cal C}$, in terms of the largest gap of $\mathcal{C}^{\perp}$ is given by
\begin{equation} \label{eq:GHW_dmin_largestgap_repeat}
d \ = \ (n+1)- g^{\perp}_{k}.
\end{equation}
The desired upper bound on $d$ will be obtained by showing the corresponding lower bound on $g^{\perp}_{k}$. This lower bound on  $g^{\perp}_{k}$ will in turn, be deduced from an appropriate upper bound on the $\left(  \lceil \frac{k}{r} \rceil - 1 \right)  (\delta-1)^{th}$ GHW, $d_{\left(  \left\lceil \frac{k}{r} \right\rceil - 1 \right)  (\delta-1)}^{\perp}$, of $\mathcal{C^{\perp}}$. It will be established in the next subsection, that under the conditions of Theorem \ref{thm:information_locality},
\bea \label{eq:proof_info_locality} \nonumber
\left(  \left\lceil \frac{k}{r} \right\rceil - 1 \right)  (\delta-1) & < & n-k  , \\
d_{\left(  \left\lceil \frac{k}{r} \right\rceil - 1 \right)  (\delta-1)}^{\perp}  & \leq &   \left(  \left\lceil \frac{k}{r} \right\rceil - 1 \right)  (r + \delta-1) .
\eea
Let  $d_{\left(  \left\lceil \frac{k}{r} \right\rceil - 1 \right)  (\delta-1)}^{\perp} \ = \ s$.  Then the number of gaps in the dual that do not exceed $s$ is given by
\begin{eqnarray}
\left| \left\{g_j^{\perp} \mid  g_j^{\perp} \leq s\right\}\right| & = &  s- \left(  \left\lceil \frac{k}{r} \right\rceil - 1 \right)  (\delta-1) \\
   & & \leq \left(  \left\lceil \frac{k}{r} \right\rceil - 1 \right)  (r + \delta-1) -  \left(  \left\lceil \frac{k}{r} \right\rceil - 1 \right)  (\delta-1) \nonumber \\
 &  = & \ r\left\lceil \frac{k}{r} \right\rceil - r \ < \ k.
\end{eqnarray}

Since there are a total of $k$  gaps in the dual code ${\cal C}^{\perp}$, there must be at least an additional $k -  \left[ s- \left(  \left\lceil \frac{k}{r} \right\rceil - 1 \right)  (\delta-1) \right]$  gaps that exceed $s$ and hence the last gap in the dual, $g^{\perp}_{k}$, satisfies the lower bound:
\bea \label{eq:upper_bound_largest_gap}  \nonumber
g^{\perp}_{k} & \geq & s \ + \ k -  \left[ s- \left(  \left\lceil \frac{k}{r} \right\rceil - 1 \right)  (\delta-1) \right]  \\
&  = &  k \ + \ \left(  \left\lceil \frac{k}{r} \right\rceil - 1 \right)  (\delta-1)  .
\eea

Combining \eqref{eq:upper_bound_largest_gap} and \eqref{eq:GHW_dmin_largestgap_repeat}, we get \eqref{eq:bound_info_locality}.
\end{proof}

\subsection{Proof of \eqref{eq:proof_info_locality}} \label{sec:proof_info_locality}

We begin with a useful lemma.


\begin{lem}\label{lem:parity_rank}
Let $\mathcal{C}$ be a systematic $[n,k,d]$ linear code whose first $k$ coordinates correspond to message symbols.    Let $S$ be a subset of $[n]$ of size $s$, such that $[k] \subseteq S$. Let ${\cal P}$ denote a sub code, supported on $S$, of the dual code ${\cal C}^{\perp}$, i.e., every code symbol in every codeword in ${\cal P}$ is zero outside of $S$.  Also, let $Q= [A_{m \times k} | B_{m \times (n-k)} ]$ with $m \geq p$, be a rank $p$, $(m \times n)$ generator matrix for ${\cal P}$.    Then we must have $\text{rank} (B) =p$ and hence $s-k \geq p$.
\end{lem}

\begin{proof}  Suppose $\text{rank} (B)  < p$.  Then the row space of $Q$ would contain nonzero vectors in its row space which are supported (i.e., nonzero in) only in the first $k$ message symbol coordinates. This is not possible as this would imply a relationship amongst the message symbols of the code ${\cal C}$.  Hence $\text{rank} (B)  = p$.  We also know that the number of nonzero columns in $B$ is less than or equal to $s-k$.  It follows that $s-k \geq p$. \end{proof}

\vspace*{0.2in}

We are now ready to prove that
\bean
\left(  \left\lceil \frac{k}{r} \right\rceil - 1 \right)  (\delta-1)  & < &  n-k
\eean
and
\bean
d_{\left(  \left\lceil \frac{k}{r} \right\rceil - 1 \right)  (\delta-1)}^{\perp}  & \leq &   \left(  \left\lceil \frac{k}{r} \right\rceil - 1 \right)  (r + \delta-1) .
\eean
For $i \in [k]$, let the $i^{th}$ code (message) symbol be locally protected by a code associated to the parity check matrix $H_i$, whose support is $S_i$ of size $|S_i|=s_i \leq r+\delta-1$. Let $V_i$ denote the row space of $H_i$ and let $\nu_i$ be its dimension. Since the null space of $H_i$ must define a code whose minimum distance is greater than or equal to $\delta$, we must have that $\nu_i \geq \delta - 1, \ \forall i \in [k]$.
Let us set $\Psi  =  \cup_{i=1}^k S_i$ and $s := |\Psi|$.

 Let $a$ be the largest integer such that there exists a subset $\{V_{i_j}\}_{j=1}^a$ with the property that if
\bea \label{eq:define_a1}
W_a & = & V_{i_1}+V_{i_2}+\cdots V_{i_a},
\eea
then for every $j_0$, $1 \leq j_0 \leq a$, we have
\bea \label{eq:define_a2}
\dim (W_a) \ - \ \dim \left( \sum_{1 \leq j \leq a, j \neq j_0} V_{i_j} \right)& \geq & \delta-1 .
\eea
In other words, each subspace $V_{i_j}$ contributes at least $(\delta-1)$ to the total dimension.
Clearly, such an $a$ exists, for $a\geq 1$ is trivially true.
Without loss of generality, we reorder the indices so that $V_{i_j}=V_j, 1 \leq j \leq a$.

We next define $W_0=\{\underline{0}\}$, $\Psi_o=\phi$ and for $1 \leq i \leq a$,
\begin{equation} \label{eq:WPsi_kbyr}
\Psi_i  =  \cup_{j=1}^i S_j, \ \  W_i  =  \sum_{j=1}^i V_j
\end{equation}
\begin{equation*}
\Delta \nu_i  =  \dim( W_i ) - \dim (W_{i-1}), \ \ \Delta s_i  =  \mid \Psi_{i} \setminus \Psi_{i-1} \mid.
\end{equation*}
Clearly,
\begin{eqnarray}
\Delta \nu_i & \geq & (\delta -1) \label{eq:two_cond1} \\
\Delta s_i & \leq & (r+\delta-1) \nonumber.
\end{eqnarray}
We now examine each subspace $V_i$ for $i=a+1,a+2,\cdots,k$ in turn.  Set
\bean
\dim (V_i + W_a) - \dim(W_a) & = & \Delta \nu_i \\
\mid S_i \setminus \Psi_a \mid & = & \Delta s_i  .
\eean
Clearly we must have
\bea  \nonumber
\Delta \nu_i &  \leq & (\delta-2) \\
\Delta s_i  & \leq & \Delta \nu_i .
\eea
The second property follows since any subset of $(\delta-1)$ or less columns of each matrix $H_i$ forms a
linearly independent set.   If either $\Delta \nu_i=0$ or $\Delta s_i=0$ we can discard $V_i$ without affecting the locality property.   Let $i_0>a$
 be the first index that has not been discarded.
We reorder the indices of the remaining $V_i$ so that the indices of $V_i, 1 \leq i \leq a$ remain unchanged
and $V_{i_0}=V_{a+1}$ and set
\bean
W_{a+1} & = & W_a + V_{a+1} \\
\Psi_{a+1} & = & \Psi_a \cup S_{a+1}.
\eean
Then
\begin{eqnarray*}
\dim (W_{a+1}) - \dim(W_a) & = & \Delta \nu_{a+1}  \\
\mid \Psi_{a+1} \setminus \Psi_a \mid & = & \Delta s_{a+1} .
\end{eqnarray*}

 Continuing in this fashion with $a$ replaced by $(a+1)$, we will eventually arrive at $W_{a+b}$ and $\Psi_{a+b}$ with
\begin{eqnarray}
\dim (W_{a+i}) - \dim(W_{a+(i-1)}) & = & \Delta \nu_{a+i} \ \leq \ (\delta-2) \nonumber \\
\mid \Psi_{a+i} \setminus \Psi_{a+(i-1)} \mid & = & \Delta s_{a+i} \ \leq  \Delta \nu_{a+i} \label{eq:two_cond2},
\end{eqnarray}
for $1 \leq i \leq b \leq k - a$. \ \  Let \bean
H & = & \left[ \begin{array}{c} H_1 \\ H_2 \\ \vdots \\ H_{a+b} \end{array} \right] .
\eean
Then
\bean
\text{rank}(H) & = &  \sum_{i=1}^{a+b} \Delta \nu_{i}, \\
\text{Supp}(H) & = & \Psi_{a+b} \\
\mid \text{Supp}(H) \mid & \leq & \sum_{i=1}^a \Delta s_i + \sum_{i=a+1}^{a+b} \Delta s_{i} \\
& \leq & a(r+(\delta-1)) + \sum_{i=a+1}^{a+b} \Delta s_{i} .
\eean
We are now in a position to apply Lemma \ref{lem:parity_rank}.  This is because the row space of the matrix $H$ can be regarded as a sub code, supported on $\Psi_{a+b}$, of the dual code ${\cal C}^{\perp}$.  Hence, from Lemma \ref{lem:parity_rank}, it must be that:
\bean
\text{Supp}(H) & \geq &  k \ + \ \text{rank}(H) \\
& = & k \ + \ \sum_{i=1}^{a+b} \Delta \nu_{i}.
\eean

 From the two expressions above for the size of the support of $H$, we obtain that
\begin{eqnarray}
a(r + (\delta-1)) + \sum_{i=a+1}^{a+b} \Delta s_i  & \geq & k + a(\delta-1) + \sum_{i=a+1}^{a+b} \Delta \nu_i \label{eq:a} \\
\implies ar & \geq  & k +  \sum_{i=a+1}^{a+b} (\Delta\nu_i -\Delta s_i) \nonumber \\
 \implies  a &  \geq & \left\lceil \frac{k}{r} \right\rceil, \label{eq:a1}
\end{eqnarray}
where \eqref{eq:a} and \eqref{eq:a1} follow from \eqref{eq:two_cond1} and \eqref{eq:two_cond2}, respectively.

It follows that the $\text{rank}(H) \geq a(\delta-1) > (\delta-1)(\lceil \frac{k}{r} \rceil - 1)$. Also, since $\text{rank}(H) \leq (n-k)$, we get that
\bean
(n-k) \ > \ (\delta-1)(\lceil \frac{k}{r} \rceil - 1),
\eean and it is hence meaningful to speak of $d_{\left(  \lceil \frac{k}{r} \rceil - 1 \right)  (\delta-1)}^{\perp}$.
Since the support of each submatrix $H_i$ is $\leq (r+\delta-1)$, we have that
\bean
d_{\left(  \lceil \frac{k}{r} \rceil - 1 \right)  (\delta-1)}^{\perp}  & \leq &   \left(  \lceil \frac{k}{r} \rceil - 1 \right)  (r + \delta-1) ,
\eean
and with this, we have recovered the two inequalities appearing in \eqref{eq:proof_info_locality}.

\vspace{0.1in}

\begin{cor}\label{cor:largest_gap}
For an $(r, \delta)_i$ code $\mathcal{C}$ that achieves the bound in \eqref{eq:bound_info_locality} with equality, we have
\begin{equation}
  d_{\left(\left\lceil \frac{k}{r} \right\rceil -1 \right)(\delta-1)  + i}^{\perp} =  k + \left(\left\lceil \frac{k}{r} \right\rceil -1 \right)(\delta-1)  + i ,
\end{equation}
for $1 \leq i \leq n - k -\left((\delta-1) \left( \lceil \frac{k}{r} \rceil -1 \right)\right).$
\end{cor}
\begin{proof}
For an optimal $(r, \delta)_i$ code, the largest gap (see \eqref{eq:GHW_dmin_largestgap_repeat}) $g^\perp_k = k + (\lceil\frac{k}{r}\rceil-1)(\delta - 1)$.
Thus there are exactly $\left(k + (\lceil\frac{k}{r}\rceil-1)(\delta - 1) \right) - k = (\lceil\frac{k}{r}\rceil-1)(\delta - 1)$ generalized dual distances ${d_i^{\perp}}$ such that $d_i^{\perp} < k + (\lceil\frac{k}{r}\rceil-1)(\delta - 1)$. Hence
\begin{equation}
  d_{\left(\left\lceil \frac{k}{r} \right\rceil -1 \right)(\delta-1)  + 1}^{\perp} =  g^\perp_k + 1
\end{equation}
and the corollary follows.
\end{proof}

\subsection{Optimality of Pyramid Codes for Information Locality} \label{sec:pyramid_codes}

We will now show that for the case $\delta \leq d$, under a suitable choice of parameters, Pyramid codes\cite{HuaCheLi} achieve the bound in Theorem \ref{thm:information_locality} with equality.

Consider an $[k+d-1,k,d]$ systematic MDS code over $\mathbb{F}_q$ having generator matrix of the form
\begin{equation}
G = \left [ \begin{array}{c|c} I_{k \times k} & Q_{k \times (d-1)} \end{array} \right].
\end{equation}
We will now proceed to modify $G$ to obtain the generator matrix for an optimal$(r,\delta)_i$ code. Let $k = \alpha r + \beta, 0 \leq \beta \leq (r-1)$ and $\delta \leq d$.  We now partition $Q$ into submatrices as shown below:
\begin{equation}
Q = \left [ \begin{array}{c|c} Q_1 & \\ \vdots & Q' \\ Q_{\alpha} & \\ Q_{\alpha+1} & \end{array} \right],
\end{equation}
where
$Q_i, 1 \leq i \leq \alpha$  are matrices of size $ r \times (\delta - 1)$, $Q_{\alpha+1}$ is of size $ \beta \times (\delta-1)$ and $Q'$ is a $k \times (d-\delta)$ matrix. Consider a second generator matrix $G'$ obtained by splitting the first $(\delta -1)$ columns of $Q$ as shown below:
\begin{equation}
G' = \left [ \begin{array}{cccc|cccc|c} I_r &&&& Q_1 &&&& \\ & \ddots &&& & \ddots &&&Q' \\ & & I_r & & & & Q_{\alpha} & &\\&&&I_{\beta} &&&&Q_{\alpha+1} & \end{array} \right],
\end{equation}
Note that $G'$ is a $k \times n$ full rank matrix, where
\begin{equation} \label{eq:meets_req}
n = k+d-1 + (\left \lceil \frac{k}{r} \right \rceil - 1)(\delta - 1).
\end{equation}
Clearly, by comparing the matrices $G$ and $G^{'}$, it follows that the code, ${\cal C}$, generated by $G'$ has minimum distance no smaller than $d$.  Furthermore, $\mathcal{C}$ is an $(r,\delta)_i$ code.  Hence, it follows from \eqref{eq:meets_req} that ${\cal C}$ is an optimal $(r,\delta)_i$ code.

\subsection{The structure of an optimal $(r,\delta)_i$ code, when $r|k$} \label{sec:bound_equality}
In this section, we will assume that $r|k$. We  borrow notation and intermediate steps used in the proof of Theorem \ref{thm:information_locality}.
\begin{thm} \label{lem:step1}
 If an $[n, k, d]$ linear code $\mathcal{C}$ having information locality $(r, \delta)$ achieves the bound in \eqref{eq:bound_info_locality} with equality, then $|S_i| = r + \delta - 1, \ S_i \cap S_j = \phi, \ 1 \leq i < j \leq a$ and $\left(\mathcal{C}^{\perp}\right)^{S_i}$ is MDS, $1 \leq i \leq a$, where $a$ is as defined together by \eqref{eq:define_a1} and \eqref{eq:define_a2}.
\end{thm}

\begin{proof}
 Since, from \eqref{eq:a1}, we have $a \geq \frac{k}{r}$, we get that $\text{dim}(W_{\frac{k}{r}}) \geq \frac{k}{r}(\delta - 1) \ \text{and} \  |\Psi_{\frac{k}{r}}| \leq \frac{k}{r}(r+\delta-1)$,
where $W_{\frac{k}{r}}$ and $\Psi_{\frac{k}{r}}$ are as defined in \eqref{eq:WPsi_kbyr}. But from Corollary \ref{cor:largest_gap}, substituting $i = \delta - 1$, we get that $d^{\perp}_{\frac{k}{r}(\delta - 1)} \ = \ k +  \frac{k}{r}(\delta - 1)$.
Hence, it must be true that
\begin{eqnarray}
\text{dim}(W_{\frac{k}{r}}) & = & \frac{k}{r}(\delta - 1) \label{eq:proof1_Wkr}
\end{eqnarray}
and
\begin{eqnarray}
|\Psi_{\frac{k}{r}}| & = & \frac{k}{r}(r+\delta-1). \label{eq:proof2_Psikr}
\end{eqnarray}
Now, since $\forall i \in [a]$, $|S_i| \leq r + \delta - 1$, from \eqref{eq:proof2_Psikr}, it follows that $|S_i|  = r + \delta - 1$ and
\begin{equation} \label{eq:proof3_SiSj}
S_i \cap S_j = \phi, \ 1 \leq i < j \leq a.
\end{equation}
Combining \eqref{eq:proof1_Wkr} and \eqref{eq:proof3_SiSj}, we also get that $\text{dim}((\mathcal{C}^{\perp})^{S_i}) = \delta -1, \forall i \in [a]$. This implies that  the dual of $(\mathcal{C}^{\perp})^{S_i}$,  which is the code, $\mathcal{C}\left|_{S_i}\right.$, has dimension $|S_i| - (\delta -1) = r$.
Now, noting $\mathcal{C}\left|_{S_i}\right.$ has parameters $[r+\delta-1, r, \delta]$, it follows that $\mathcal{C}\left|_{S_i}\right.$ and hence $(\mathcal{C}^{\perp})^{S_i}$ are MDS, $\forall i \in [a]$.
\end{proof}
\vspace{0.1in}

\begin{thm} \label{thm:necessary_conditions_info_locality}
If an $[n, k, d]$ linear code $\mathcal{C}$ having information locality $(r, \delta)$ achieves the bound in \eqref{eq:bound_info_locality} with equality and $d < r+2\delta-1$, then $\delta \leq d$ and up to a reordering of columns, the parity check matrix, $H$ of $\mathcal{C}$ can be assumed to be of the form:
\begin{equation} \label{eq:parity_structure_equality}
H  =  \left[ \begin{array}{ccc|ccc|c} Q_1 &&& I_{\delta-1} &&&\\  & \ddots & && \ddots & &0  \\ & & Q_{(\frac{k}{r})} &&& I_{\delta-1} & \\  \hline \\& A & && 0 && I_{d-\delta} \end{array} \right],
\end{equation}
where \begin{equation}
A = \left[A_1 \mid A_2 \mid \ldots \mid A_{\frac{k}{r}}\right ]
\end{equation}
and $\forall i \in \left[ \frac{k}{r}\right]$, the matrix
 \begin{equation}
 \left[ \begin{array}{c|c|c} Q_i & I_{\delta-1} & 0\\ A_i & 0 & I_{d-\delta} \end{array} \right]
 \end{equation}
generates an $[r+d-1, d-1, r+1]$ MDS code.
The matrices $Q_i$ and $A_i$ appearing above are of sizes $(\delta-1) \times r$ and $(d-\delta) \times r$ respectively.
\end{thm}

\begin{proof}
 We will prove Theorem \ref{thm:necessary_conditions_info_locality} in two steps. For an optimal $(r, \delta)_i$ code, we will show that
\begin{itemize}
\item{\underline{Step $1$}}: $a = \frac{k}{r}$ and $\delta \leq d$. These, along with Theorem \ref{lem:step1} will directly mean that the matrix $H$ has the form given in \eqref{eq:parity_structure_equality}.
\item{\underline{Step $2$}}: Secondly, we will show that $\forall i \in \left[ \frac{k}{r}\right]$, the matrix
\begin{equation}
\left[ \begin{array}{c|c|c} Q_i & I_{\delta-1} & 0\\ A_i & 0 & I_{d-\delta} \end{array} \right]
\end{equation}
generates an $[r+d-1, d-1, r+1]$ MDS code.
\end{itemize}

\vspace{0.1in}

\underline{Proof of Step $1$}:
Let $a \geq \frac{k}{r} + 1$. Then, from Theorem \ref{lem:step1} it follows that
\begin{eqnarray}
 n & \geq & a (r+\delta-1) \label{eq:d_delta}\\
   & \geq & (\frac{k}{r} + 1) (r+\delta-1) \nonumber  \\
   &   =  & k + \frac{k}{r}(\delta -1) + r +\delta -1 \nonumber \\
  &  > & k +d - \delta + \frac{k}{r}(\delta -1) \label{eq:a_equal_kbyr},
\end{eqnarray}
where \eqref{eq:a_equal_kbyr} follows from the assumption that $d < r+2\delta-1$. But \eqref{eq:a_equal_kbyr} contradicts the assumption the code is optimal (see \eqref{eq:bound_info_locality}) and hence $a = \frac{k}{r}$.

Next, in order to show that $\delta \leq d$, first note from \eqref{eq:d_delta} that the length of an optimal $(r, \delta)_i$ code must be at least $\frac{k}{r}(r + \delta - 1)$. But, if one assumes $\delta  > d$, then from \eqref{eq:bound_info_locality}, we get that, under optimality,
\begin{eqnarray*}
n & = & d + k - 1 + \left(\frac{k}{r} - 1\right)(\delta - 1) \\
  &  <  & \delta + k - 1 + \left(\frac{k}{r} - 1\right)(\delta - 1) \\
  & = & \frac{k}{r}(r + \delta - 1),
\end{eqnarray*}
which results in a contradiction. Hence we conclude that, under optimality, $\delta \leq d$.

\vspace{0.1in}

\underline{Proof of Step $2$}:
From Theorem \ref{lem:step1} and step $1$, we get that the parity check matrix, $H$, for the code $\mathcal{C}$ has the form (up to permutation of columns) given in \eqref{eq:parity_structure_equality}. Equivalently, the generator matrix, $G$, of $\mathcal{C}$, which, up to a permutation of columns is of the form
\begin{equation} \label{eq:generator_structure_equality}
G  =  \left[ \begin{array}{ccc|ccc|c} I_r &&& Q_1^t &&& A_1^t\\  & \ddots &&& \ddots & &\vdots  \\ &&I_r &&& Q_{\frac{k}{r}}^t & A_{\frac{k}{r}}^t \end{array} \right].
\end{equation}
Let $T$ denote the index set for the last $d - \delta$ columns of $G$ (i.e., the columns corresponding to $A_i$s) and consider a shortened code $\mathcal{C}^S$ of $\mathcal{C}$, where the $S = S_1 \cup T$. Note that $\mathcal{C}^S$ is generated by the matrix $G_S = [I_r \mid Q_1^t\mid A_1^t]$ and hence $\mathcal{C}^S$  is an $[r + d - 1, r, d_S]$ code, where $d_S$ denotes the minimum distance of $\mathcal{C}^S$. Clearly, $d_S \geq d$, the minimum distance of the code $\mathcal{C}$, since shortening a code only increases the minimum distance. This means that  $\mathcal{C}^S$  has parameters $[r + d - 1, r, d]$, i.e., $\mathcal{C}^S$ is MDS and so is its dual.
\end{proof}
\vspace{0.1in}

\begin{cor}
If $r|k$, $d < r+2\delta-1$ and equality is achieved in \eqref{eq:bound_info_locality}, then
    \begin{equation}
    d^{\perp}_{i} = r + i \ \ \ \ \ 1 \leq i \leq \delta-1.
    \end{equation}
\end{cor}

\vspace{0.1in}


%

\vspace{0.1in}

\section{Codes with All-Symbol Locality} \label{sec:all_symbol_locality}

In this section, we study $(r,\delta)_a$ codes for the case when $(r+\delta-1)|n$ and $\delta \leq d$.
Firstly, for the case when $n=\lceil \frac{k}{r} \rceil  (r+\delta-1)$, we will give an explicit construction of a code with all-symbol locality by splitting this time, rows of the parity check matrix of an appropriate MDS code.  We will refer to this as the parity-splitting construction.
The code so obtained is optimal with respect to \eqref{eq:bound_info_locality}.
We will also show the existence of optimal codes with all-symbol locality
without the restriction $n=\lceil \frac{k}{r} \rceil  (r+\delta-1)$. The proof of this theorem uses random coding arguments similar to those used for proving Theorem 17 in \cite{GopHuaSimYek}.

\subsection{Explicit and Optimal $(r,\delta)_a$ Codes via Parity-Splitting}
\begin{thm}
 Let $n=\lceil \frac{k}{r} \rceil  (r+\delta-1)$ and $\delta \leq d$. Then, for $q > n$, there exists an explicit and optimal $(r,\delta)_a$ code over $\mathbb{F}_q$.
\end{thm}
\begin{proof}
 Let $H'$ be the parity check matrix of an $[n,k',d]$ Reed-Solomon code over $\mathbb{F}_q$, where
$k'= k+ (\lceil \frac{k}{r} \rceil  - 1)(\delta-1)$ and
$d = n-k'+1  = n-k+1- (\lceil \frac{k}{r} \rceil - 1)(\delta-1)$.  Such codes exist if $q > n$.
We choose $H'_{(n-k') \times n}$ to be a Vandermonde  matrix.  Let
\begin{equation}
H'= \left[ \begin{array}{c}
     Q_{(\delta-1) \times n} \\
     A_{(n-k'+1-\delta)\times n}
    \end{array} \right].
\end{equation}
 We partition the matrix $Q$ in terms of submatrices as shown below
\begin{equation}
Q = \left [ Q_1 \mid Q_2 \mid \ldots \mid Q_{\lceil \frac{k}{r} \rceil} \right],
\end{equation}
where
$Q_i, 1 \leq i \leq {\lceil \frac{k}{r} \rceil}$  are matrices of size $ \delta-1 \times (r+\delta - 1)$.
Next consider the code
$\mathcal{C}$ whose parity check matrix, $H$, is obtained by splitting the first $\delta -1$ rows of $H'$ as follows:
\begin{equation}
 H=\left[ \begin{array}{ccc}
            Q_1 &&\\
	    & \ddots &\\
	    &&Q_{\lceil \frac{k}{r} \rceil}\\
	    \hline
	    & A &
           \end{array}
\right].
\end{equation}
Due to the Vandermonde structure of $H'$, all rows of $H$ are linearly independent. Thus
 $\text{Rank}(H)=n-k'+(\lceil \frac{k}{r} \rceil - 1)(\delta -1)$. Thus
$\text{dim}(\mathcal{C})= k' - (\lceil \frac{k}{r} \rceil - 1)(\delta -1) = k$.
It is also clear from the construction that this code is an $(r,\delta)_a$ code.

Let $d_{min}$ be the minimum distance of $\mathcal{C}$. Since any set of columns\footnote{set here indicates indices of the columns} of $H$
which are linearly dependent are also linearly dependent in $H'$,
$d_{min} \geq d =  n-k+1- (\lceil \frac{k}{r} \rceil - 1)(\delta-1) $.
But, by \eqref{eq:bound_info_locality}, we must have $d_{min}\leq d$. Hence $d_{min}=d$.
\end{proof}
\vspace{0.1in}
\begin{note}
 In the above construction, let $k=\alpha r+\beta$. Let $\delta_k = r-\beta$.
Then
\begin{eqnarray}
 d-\delta &=& n-k - (\lceil \frac{k}{r} \rceil )(\delta-1) \nonumber \\
	  &=& r\frac{k+\delta_k}{r}-k \nonumber \\
	  &=& \delta_k.
\end{eqnarray}
In particular if $r|k$, $\delta_k=0$ and hence $d=\delta$.
\end{note}

\subsection{Existence of Optimal $(r,\delta)$ codes with All-Symbol Locality}
Here, we will state a couple of definitions and a lemma from \cite{GopHuaSimYek}, which will be useful in proving the existence of optimal codes with all symbol locality.

\vspace{0.1in}

\begin{defn}[$k$-core \cite{GopHuaSimYek}]
Let $L$ be a subspace of $\mathbb{F}_q^n$ and $S \subseteq [n]$ be a set of size $k$. $S$ is said to be a $k$-core
for $L$ if for all vectors $v \in L$, $\text{Supp}(\mathbf{v}) \nsubseteq S$.
\end{defn}

\vspace{0.1in}

$S$ is a $k$-core of a linear code $\mathcal{C}$ if and only if the $k$ columns of the generator matrix of the dual code $\mathcal{C}^{\perp}$ corresponding to $S$ are linearly independent.

\vspace{0.1in}

\begin{defn}[Vectors in General Position Subject to $L$ \cite{GopHuaSimYek}]
 Let $L$ be a subspace of $\mathbb{F}_q^n$. Let $G = [\mathbf{g}_1, \cdots , \mathbf{g}_n ]$ be a $k \times n$ matrix over $\mathbb{F}_q$. The columns of $G$, $\{\mathbf{g}_i \}_{i=1}^n$ are said to be in general position subject to $L$ if:
\begin{itemize}
 \item Row space of $G$, denoted by Row$(G) \subseteq L^\perp$.
 \item For all $k$-cores $S$ of $L$, we have $\text{Rank}(G|_S ) = k$.
\end{itemize}
\end{defn}

\vspace{0.1in}

\begin{lem}[Lemma 14 of \cite{GopHuaSimYek}] \label{exist}
 Let $n,k,q$ be such that $q > kn^k$. Let $L$ be a subspace of $\mathbb{F}_q^n$ and $ 0 < k \leq n - \text{dim}(L)$.
 Then $\exists$ a set of vectors $\{\mathbf{g}_i\}_{i=1}^n$ in $\mathbb{F}_q^k$
that are in general position subject to $L$.
\end{lem}

\vspace{0.1in}

Using the above lemma, we will now prove the existence of optimal $(r,\delta)$ codes for a general set of parameters.

\vspace{0.1in}

\begin{thm}
 Let $q > kn^k$, $(r+\delta-1)|n$ and $\delta \leq d$. Then there exists an optimal $(r,\delta)_a$  code over $\mathbb{F}_q$
\end{thm}
\begin{proof}
Proof is similar to the proof of Theorem $17$ of \cite{GopHuaSimYek}. Let $n = (r+\delta-1)t$. Let $\{P_1, \cdots, P_t\}$ be a partition of $[n]$, where $|P_i| = r+\delta-1, 1\leq i \leq t$. Let $Q_i$ be the parity check matrix of an $[r+\delta-1,r,\delta]$ MDS code with support $P_i$. Consider
\begin{equation}
H'_{t(\delta -1) \times n}
	      = \begin{bmatrix}
               Q_1 &    &  &  \\
	          & Q_2 &  &  \\
	        &   & \ddots &  \\
	          &    &  &Q_t  \\
             \end{bmatrix}.
\end{equation}

Let $L=\text{Rowspace}(H')$. Since $\delta \leq d$, \eqref{eq:bound_info_locality} gives that
\begin{equation}
 n-k  \geq  \left\lceil \frac{k}{r} \right\rceil(\delta -1) \geq \frac{k}{r}(\delta -1) \nonumber \\
\end{equation}
Rearranging the above equation, we get
\begin{equation}
k \leq n - t(\delta-1)
\end{equation}
which implies the existence of $k$-cores for $L$ exist. Thus, from Lemma \ref{exist},  $\exists \ \{\mathbf{g}_i\}_{i=1}^{n}$, $\mathbf{g}_i \in \mathbb{F}_q^k$ which are in general position subject to $L$. Now consider the code $\mathcal{C}$ whose generator matrix $G_{k\times n}=[\mathbf{g}_1 \cdots \mathbf{g}_n]$. Clearly, $\mathcal{C}$ is an $(r,\delta)_a$ code, whose length is $n$ and dimension is $k$. It remains to prove that $d_{min}(\mathcal{C}) = d$, given by the equality condition in \eqref{eq:bound_info_locality}. Towards this, we will show that (see next sub section) for any set $S \subseteq [n]$ such that $\text{Rank}(G|_S) \leq k-1$, it must be true that
\begin{equation} \label{eq:S_cardinality_allsymbloc}
 |S| \leq k-1+(\delta-1)\left( \displaystyle \lceil \frac{k}{r} \rceil -1 \right).
\end{equation}
Now the minimum distance of $\mathcal{C}$ is given by
\begin{equation}
 d= n- \max_{\substack{S \subseteq [n] \\ \text{Rank}(G|_S) \leq k - 1}} |S| \geq
n-k+1 - \left( \displaystyle \lceil \frac{k}{r} \rceil -1 \right)(\delta-1).
\end{equation}
Combining the above equation and \eqref{eq:bound_info_locality}, it follows that the code ${\cal C}$ has the distance given in the theorem statement.
\end{proof}

\vspace{0.1in}

\subsubsection{Proof of \eqref{eq:S_cardinality_allsymbloc}}

Let $S \subseteq [n]$ be such that $\text{Rank}(G|_S) \leq k-1$. Clearly, $S$ does not contain a $k$-core. Also, note that any $K \subseteq [n], |K|=k $, is a $k$-core for $L$ if and only if
\begin{equation}
|P_i \cap K| \leq r \ \ \forall \ i  \in [t].
\end{equation}
Thus there exists $\text{some } i \in [t]$ such that $|P_i \cap S| \geq r+1$.

Define
\begin{equation*}
 b_{\ell} := \left|\left\{ i \in [t] \arrowvert \  |P_i \cap S| = r+\ell \right\}\right| \ \ \ 1 \leq \ell \leq \delta -1.
\end{equation*}
For $1 \leq \ell \leq \delta -1$, consider the set, $S_{\ell}$, obtained from $S$ by dropping $\ell$ elements of $S$ from each of the $b_{\ell}$ sets $\{P_i | \ |P_i \cap S| = r + \ell\}$. Clearly, the set $\cup_{1 \leq \ell \leq \delta -1}S_{\ell}$ is an $|S|-b_1-2b_2-\cdots-(\delta-1)b_{\delta-1}$ core and thus
\begin{equation} \label{eq:nc1}
 |S|-(\delta-1)(\sum_{i=1}^{\delta-1} b_i)  \leq |S|-b_1-2b_2-\cdots-(\delta-1)b_{\delta-1} \leq k-1.
\end{equation}
Also if we pick $r$ co-ordinates from each $P_i$ such that $|P_i \cap S| \geq r + 1$, we get a $(r)(\sum_{i=1}^{\delta-1} b_i)$-core.
Thus,
\begin{equation} \label{eq:nc2}
 \sum_{i=1}^{\delta-1} b_i \leq \left\lfloor \frac{k-1}{r} \right\rfloor = \left \lceil \frac{k}{r} \right \rceil -1.
\end{equation}
 Combining \eqref{eq:nc1} and \eqref{eq:nc2}, we have
\begin{equation}
 |S| \leq k - 1 + \left(  \left \lceil \frac{k}{r} \right \rceil -1 \right)(\delta-1).
\end{equation}

\vspace{0.1in}

\section{An upper bound on the minimum distance of concatenated codes} \label{sec:concatenated_codes}

Consider a (serially) concatenated code (see \cite{For}, \cite{Dum}) having an $[n_1, k_1, d_1]$ code $\mathcal{A}$ as the inner code and an $[n_2, k_2, d_2]$ code $\mathcal{B}$ as the outer code.  Clearly, a concatenated code falls into the category of an $(r,\delta)_a$ code with $\delta = d_1$, $r = n_1 - d_1 + 1$.   Hence, the bound in \eqref{eq:bound_locality} applies to concatenated codes as well. Using the fact that a concatenated code has length $n = n_1n_2$, dimension $k = k_1k_2$, we obtain from the results of the present paper, the following upper bound on minimum distance $d$:
\begin{equation} \label{eq:bound_concatenated}
d \ \leq \ n_1n_2 - k_1k_2  +1 -\left(\left\lceil{\frac{k_1k_2}{n_1-d_1+1}}\right\rceil - 1\right)(d_1 - 1 ).
\end{equation}
Well known bounds on the minimum distance of a concatenated codes are
\begin{equation} \label{eq:distance_concatenated_upperbound_known}
d_1d_2 \leq d \ \leq \ n_1d_2.
\end{equation}
In practice, concatenated codes often employ an interleaver between the inner and outer codes in order to increase the minimum distance \cite{BenDivMonPol}.  In this case, while the upper bound in \eqref{eq:distance_concatenated_upperbound_known} no longer holds, the bound in \eqref{eq:bound_concatenated} continues to hold.

An asymptotic version of \eqref{eq:bound_concatenated} can be obtained, if we assume that both the component codes are MDS. Let $R = \frac{k}{n}$ and $\Delta = \frac{d}{n}$, respectively, denote the rate and the fractional distance of the concatenated code. Similarly let $R_1, \Delta_1, R_2, \Delta_2$ denote the corresponding parameters of the component codes. We consider the asymptotic case when both $n_1$ and $n_2$ tend to infinity. Also assume that $d_i, k_i$ increase in proportion with $n_i, \ i = 1, 2$. Using all these fractional parameters in \eqref{eq:bound_concatenated}, we get
\begin{equation} \label{eq:eq:bound_concatenated_asymptotic_temp}
\Delta \ \leq 1 - R_1R_2 - \Delta_1R_2 + \frac{R_2}{n_1} + \frac{d_1 - 1}{n_1n_2}.
\end{equation}
In the limit $n_i \rightarrow \infty, \ i = 1, 2$, singleton bound gives us $R_i = 1 - \Delta_i$. Using this fact, asymptotically \eqref{eq:eq:bound_concatenated_asymptotic_temp} becomes
\begin{equation} \label{eq:bound_concatenated_asymptotic}
\Delta \ \leq 1 - R_2 = \ 1 - \frac{R}{R_1}.
\end{equation}
We remark that this is the same asymptotic bound that one would get from \eqref{eq:distance_concatenated_upperbound_known}.

\bibliographystyle{IEEEtran}
\bibliography{Locality_bib}

\end{document}